\documentclass{llncs}
\usepackage[utf8]{inputenc}
\usepackage{amssymb, amsmath, tikz}
\newcommand{\B}{\{0,1\}}

\newcommand{\SC}{\mathcal{S}}
\newcommand{\SCC}{\mathcal{O}}
\newcommand{\FP}{\operatorname{FP}(f)}
\newcommand{\powerset}[1]{2^{#1}}
\newcommand{\card}[1]{|{#1}|}
\newcommand{\gp}[1]{G[P_#1 ]}
\usepackage{hyperref}
\tikzset{main node/.style={circle,fill=white,draw,minimum size=0.5cm,inner sep=0pt}, outer sep=0.5mm		}

\begin{document}

\title{Relationship between the Reprogramming Determinants of Boolean Networks and their Interaction Graph}
\author{Hugues Mandon\inst{1}, Stefan Haar\inst{2}, Lo\"ic Paulev\'e\inst{1}}
\institute{%
LRI UMR 8623, Univ. Paris-Sud – CNRS, Universit\'e Paris-Saclay, France
\and
LSV, ENS Cachan, INRIA, CNRS, Universit\'e Paris-Saclay, France
}

\maketitle

\begin{abstract}
In this paper, we address the formal characterization of targets triggering cellular trans-differentiation
in the scope of Boolean networks with asynchronous dynamics.
Given two fixed points of a Boolean network, we are interested in all the combinations of mutations
which allow to switch from one fixed point to the other, either possibly, or inevitably.
In the case of existential reachability, we prove that the set of nodes to (permanently) flip are only and 
necessarily in certain connected components of the interaction graph.
In the case of inevitable reachability, we provide an algorithm to identify a
subset of possible solutions.
\end{abstract}

\section{Introduction}

In the field of regenerative medicine, an emerging way to treat patients is to
reprogram cells, leading, for instance, to tissue or neuron regeneration. Such a
challenge has become realistic after first experiments have shown that some of
the cell fate decisions can be reversed \cite{Takahashi2016}. Whereas the cells
go through several multipotent states before reaching a differentiated state,
the differentiation process can be inversed,  producing induced pluripotent stem
cells (iPSCs) from an already differentiated cell. By using a distinct
differentiation path, this allows to "transform" the type of a cell.
Alternatively, it is also possible to directly perform a trans-differentiation
without necessarily going (back) through a multipotent state
\cite{Graf2009,delSol14}.

In the aforementioned work, the de- and trans-differentiation has been achieved
by targeting specific genes, that we refer to as \emph{Reprogramming
Determinants} (RDs), through the mediation of their transcription factors
\cite{Takahashi2016,Crespo2013}.

The computational prediction of RDs requires to assess multiple features of the cell dynamics and the reprogramming strategy, such as the impact of the kind of perturbations (persistent versus temporary) and of their order; the nature of targeted cell type (differentiated/pluripotent), and the desired  inevitability of their reachability (fidelity); the nature and duration of the triggered cascade of regulations (efficiency); and finally, the RD robustness with respect to initial state heterogeneity among cell population, and with respect to uncertainties in the computational model.

So far, no general framework allows to efficiently encompass those features to systematically predict best combinations of RDs in distinct cellular reprogramming events.

In this paper, we address the identification of RDs from \emph{Boolean Networks} (BNs) which model the dynamics of gene regulation and signalling networks. The state of the components (or nodes) of the networks are represented by Boolean variables, and the state changes are specified by Boolean functions which associate the next state of nodes, given the (binary) state of their regulators
\cite{Thomas73,Aracena2008}.
BNs are well suited for an automatic reasoning on large biological networks where the available knowledge is mostly about activation and inhibition relations\cite{Abou-Jaoude2015}. Such activation/inhibition relations between components form a signed directed graph, that we refer to as the \emph{Interaction Graph}. 

In this work, we make the assumption that the differentiated cellular states correspond to the \emph{attractors} of the dynamics of the computational model, i.e., the long-run behaviours.
In the scope of BNs, those attractors can be of two kinds: either a single state (referred to as a fixed point), or a terminal cyclic behaviour.

The relationship between the IG of BNs and the number of their attractor has been extensively studied
\cite{Aracena2008,Richard09-MaxFP,Richard10-AAM}.
However, little work exists on the characterization of the perturbations which trigger a change of attractor.
Currently, most of RDs prediction are performed using statistical analysis on
expression data in order to rank candidate transcription factors
\cite{Chang2011,Rackham2016,Jo2016}.
Whereas based on network models, those approaches do not allow to derive a
complete set of solution for the reprogramming problem.
In \cite{Crespo2013}, the authors developed a heuristic to derive candidate RDs
from a pure topological analysis of the interaction graph:
the RDs are selected only in positive cycles that have different values in the
started and target fixed points.
However, there is no guarantee that the derived RDs can actually lead to a
change of attractor in the asynchronous dynamics of the Boolean networks, and
neither that the target fixed point is the only one reachable.
Finally, \cite{Gao2016} gives a formal characterization of RDs subject to
temporal mutations which trigger a change of attractor in the synchronous
semantics of conjunctive Boolean networks.

\paragraph{Contribution}
This work relies on model checking and reachability analysis, that have been proved useful and successful in previous studies\cite{Abou-Jaoude2015,P16-CMSB}.

Given a BN, all of whose attractors are fixed points, given an initial fixed point and a target fixed point, we provide a characterization of the candidate RDs (set of nodes) with respect to the interaction graph and for two settings of cellular reprogramming:
\begin{itemize}
\item 
with a permanent perturbation of RDs, the target fixed point becomes reachable in the asynchronous dynamics of the BN;
\item
with a permanent perturbation of RDs, the target fixed point is the sole reachable attractor in the asynchronous dynamics of the BN.
\end{itemize}
For the first case, we prove that all the RDs are distributed among particular strongly connected components of the interaction graph, and we give algorithms to determine them in both settings.
In the second case, we prove that only some of them are distributed among
strongly connected components of the interaction graph.
We provide an algorithm to identify possible combination of permanent
perturbations leading to inevitable reachability of the target fixed point.
Whereas the algorithm may miss some solutions, all returned solutions are
correct.

\paragraph{Outline}
Section ~\ref{sec:background} gives the definitions and basic properties of BNs and of their asynchronous dynamics. The formalization of the BN reprogramming problem with permanent perturbations of nodes is established in Sect.~\ref{sec:bn-reprogramming}. Section ~\ref{sec:rd-scc} states the main results on the characterization of RDs with respect to the interaction graph of BNs. An algorithm to enumerate all RDs by exploiting this characterization is given in Sect.~\ref{sec:enumeration}. Finally, \ref{sec:discussion} discusses the results and sketches future work.

\subsection*{Notations}

Given a finite set $I$, $\powerset{I}$ is the power set of $I$, $\card{I}$ the
cardinality.
Given a positive integer $n$, $[n] = \{1,\dots,n\}$.

Given a Boolean state $x\in\B^n$ and set of indexes $I\subset [n]$,
$\bar{x}^I$ is the state where $\bar{x_i}^I = x_i$ if $i \notin I$ and $\bar{x_i}^I = 1-x_i$ if $i \in I$.
Similarly, given $x,y\in\B^n$,  $x_{[x_{I} = y_{I}]}$ denotes the state where
for all $i \in I,\, (x_{[x_{I} = y_{I}]})_i = y_i$ and for all $i \notin I,\,
(x_{[x_{I} = y_{I}]})_i = x_i $\\

\section{Background}
\label{sec:background}

In this section, we give the formal definition of Boolean networks, their
interaction graph and transition graph in the asynchronous semantics.
Finally, we recall the main link between their attractors and the positive
cycles in their interaction graph.

\subsection{Definitions}

\paragraph{Boolean Network (BN):}
A BN is a finite set of Boolean variables, each of them having a Boolean function. This function is a logical Boolean function depending from the network's variables and determining the next state of the variable.
\begin{definition}[Boolean Network (BN)]
A Boolean Network is a function $f$ such that:
\begin{align*}
f&: &\{0,1\}^n & \rightarrow \{0,1\}^n \\
& &x=(x_1,...,x_n) & \mapsto f(x) = (f_1(x),...,f_n(x))
\end{align*}
\label{def:bn}
\end{definition}

\begin{example}
An example of BN of dimension $3$ ($n = 3$) is
\begin{align*}
f_1(x) &=x_3 \vee (\neg x_1 \wedge x_2) \\
f_2(x) &= \neg	 x_1 \vee x_2 \\
f_3(x) &= x_3 \vee (x_1 \wedge \neg x_2)
\end{align*}
\label{ex:baseExample}
\end{example}

\paragraph{Interaction Graph:} 
To determine the RDs, we rely on a simplification of the interactions between
the genes, and of the concentrations. A gene will either be active or inhibited.
Gene interactions are simplified likewise, a gene either activates or inhibits
another gene, and we ignore time scales. With this in mind, an \emph{interaction graph} (Def.\ref{def:ig}) can be build: genes are the vertices, and the interactions are the oriented arcs, labelled either $+$ or $-$, if it is an activation or an inhibition.

\begin{definition}[Interaction Graph]
An interaction graph is noted as $G = (V, E)$, with $V$ being the vertex set, and $E$ being the directed, signed edge set, $E \subset (V \times V \times \{-,+\})$
\label{def:ig}
\end{definition}

A cycle between a set of nodes $C\subseteq V$ is said positive (resp. negative) if and only if
there is an even (odd) number of negative edges between those nodes.

An interaction graph can also be defined as an abstraction of a Boolean network: the functions are not given and not always known, but if a vertex $u$ is used in the function $f_v$, there is an edge from $u$ to $v$, negative if $f_v(x)$ contains $\neg x_u$ and positive if it contains $x_u$.

\begin{definition}[Interaction Graph of a Boolean network ($G(f)$)]
An interaction graph can be obtained from the Boolean network $f$:
the vertex set is $[n]$, and for all $u,v \in [n]$ there is a positive (resp. negative) arc from $u$ to $v$ if $f_{vu}(x)$ is positive (resp. negative) for at least one $x \in \{0,1\}^n$ (For every $u,\, v \in \{1, . . . , n\}$, the function $f_{vu}$ is the discrete derivative of $f_v$ considering $u$,  defined on $\{0, 1\}^n$ by : $f_{vu}(x):=f_v(x_1,..,x_{u-1},1,x_{u+1},..,x_n) - f_v(x_1,..,x_{u-1},0,x_{u+1},..,x_n)$).
\end{definition}

Given an interaction graph $G=(V,E)$, and one of its vertex $u \in V$, $P_u$
denotes the set of ancestors of $u$, i.e., the vertices $v$ for which there exists a path in $E$ from $v$ to $u$.
Similarly, $p_u$ is the set of the parents of $u$, i.e., $v \in p_u \Rightarrow (v,u,s) \in E$.
Furthermore, $\gp{u}$ is the induced subgraph of $G$ with $P_u$ as vertex set.

Fig.~\ref{fig:example-ig} gives an example of an interaction graph, which is also
equal to $G(f)$, where $f$ is the Boolean network of Ex.\ref{ex:baseExample}.
\begin{figure}[t]
\begin{center}
\begin{tikzpicture}[node distance=2cm, every loop/.style={}]
\node[main node]		(2)		{2} ;
\node[main node]		(1)	[below left of=2]	{1} ;
\node[main node]		(3)	[below right of=2]	{3} ;
\path	(1)	edge[-|, bend left, red]		node[left] {$-$}	(2)
		(1)	edge[-|, loop left, red, ] 	node {$-$}	(1)
		(1) edge[->, bend right, blue]	node[below] {$+$}	(3)
		(2) edge[->, bend left, blue] 	node[right] {$+$}	(1)
		(2) edge[->, loop above, blue] 	node {$+$}	(2)
		(2) edge[-|, red]				node[right] {$-$}	(3)
		(3) edge[->, blue]				node[above] {$+$}	(1)
		(3)	edge[->, loop right, blue]	node {$+$}	(3) ;
		
\end{tikzpicture}
\end{center}
\caption{Interaction graph of Ex.\ref{ex:baseExample}
A "normal" blue arrow means an activation, and a "flattened" red arrow means an inhibition.
\label{fig:example-ig}
}
\end{figure}
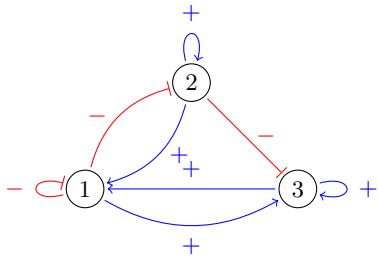

\paragraph{Transition Graph:}
We model the dynamics of a Boolean network $f$ by \emph{transitions} between its
states $x\in\B^n$.
In the scope of this paper, we consider the \emph{asynchronous semantics} of
Boolean networks:
a transition updates the value of only one vertex $u\in[n]$.
From a single $x\in\B^n$, one  has different transitions for each vertex $u$
such that $f_u(x) \neq x_u$.
This leads to the definition of the \emph{transition graph}
(Def.~\ref{def:transition-graph}) where vertices are
all the possible states $\B^n$, and edges correspond to asynchronous
transitions.
\begin{definition}[Transition graph]
\label{def:transition-graph}
The transition graph is the graph having $\{0,1\}^n$ as vertex set and the edges
set $\{x \to \bar{x}^{\{u\}} \mid x\in\B^n, u\in[n], x_u \neq (f(x))_u\}$. An existing path from $x$ to $y$ is noted $x \to^* y$.
\end{definition}

Fig.\ref{fig:ex-tg} gives the transition graph of the asynchronous dynamics of
Boolean network of Ex.\ref{ex:baseExample}.
\begin{figure}
\begin{center}
\begin{tikzpicture}[node distance=1.5cm]
\node[draw]	(000)						{000} ;
\node[draw]	(010)	[above of=000]		{010} ;
\node[draw]	(100)	[right of=000]		{100} ;
\node[draw]	(110)	[right of=010]		{110} ;
\node[draw]	(001)	[above right of=100]	{001} ;
\node[draw]	(011)	[above of=001]		{011} ;
\node[draw]	(101)	[right of=001]		{101} ;
\node[draw]	(111)	[right of=011]		{111} ;

\node[draw, align=left, minimum width=2.5cm, minimum height=0.8cm] at (0.8,1.5) [magenta] {} ;
\node[draw, align=left, minimum width=0.9cm, minimum height=0.7cm] at (4.05,1.05) [magenta] {} ;
\node[draw, align=left, minimum width=0.9cm, minimum height=0.7cm] at (4.05,2.55) [magenta] {} ;

\path	(000)	edge[->]		(010)
		(001)	edge[->]		(011)
		(001)	edge[->]		(101)
		(010)	edge[bend left, ->]	(110)
		(011)	edge[->]		(111)
		(100)	edge[->]		(000)
		(100)	edge[->]		(101)
		(110)	edge[bend left, ->]	(010) ;
\end{tikzpicture}
\end{center}
\caption{Transition graph of the Boolean network defined in Ex.\ref{ex:baseExample}.
We use shorter notations, $010$ meaning that the node $1$ has $0$ as value, the node $2$ has $1$ as value, and the node $3$ has $0$ as value.
The attractors are boxed in magenta.
\label{fig:ex-tg}}
\end{figure}
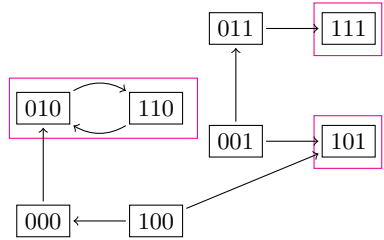


\paragraph{Attractors, Fixed point :}
BN's \emph{Attractors} are the terminal strongly connected components of the transition graph, and can be seen as the long-term dynamics of the system. Note that an attractor is always a set of states, but it can contain either multiple distinct nodes, that is the system oscillate between multiple states (\emph{cyclic attractor}) or a unique point, i.e the system stays in the same state (\emph{fixed point}).
\begin{definition}[Attractor]
\begin{align}
S \subseteq \{0,1\}^n \mbox{ is an attractor} \Leftrightarrow & \, S \neq \emptyset \\
&\mbox{ and } \forall x \in S,\, \forall y \in \{0,1\}^n \setminus S,\, x \not \to y \\
&\mbox{ and } \forall x \in S,\, S \setminus x \mbox{ does not verify (2) }
\end{align}
If $ \card S=1$ then $S$ is a fixed point. Otherwise $S$ is a cyclic attractor.
\end{definition}

Given a BN $f$, $\FP \subseteq \{0,1\}^n$ denotes the set of its fixed points
($\forall x \in \FP, f(x)=x$).

\begin{example}
The BN of Ex.\ref{ex:baseExample} has 3 attractors that correspond to the 3
terminal strongly connected components of Fig.\ref{fig:ex-tg}:
$\{010,110\}$ (cyclic attractor), $\{101\}$ and $\{111\}$ (fixed points).
\end{example}

\subsection{On the link between attractors and the interaction graph}

Theorem~\ref{thm:thomas} is a conjecture by Ren\'e Thomas \cite{Thomas73} that
has been since demonstrated for Boolean and discrete networks
\cite{Aracena2008,REMY2008335}:
if a Boolean network has multiple attractors then its interaction graph
necessarily contains a positive cycle.
In the case of multiple fixed points, any pair of fixed point differ at least on
a set of nodes forming a positive cycle.

\begin{theorem}[Thomas' first rule] 
If $G=(V,E)$ has no positive cycles, then $f$ has at most one attractor. Moreover,
if $f$ has two distinct fixed points $x$ and $y$, then $G$ has a positive
cycle between vertices $C\subseteq V$ such that $x_v \neq y_v$ for every vertex $v$ in $C$.
\label{thm:thomas}
\end{theorem}

We can also remark that for a vertex to stay at a value $y_v$ where $y$ is a fixed point, it only needs its ancestors to have the same values as in $y$.

\begin{remark}
$\forall y \in \FP,\,\forall u \in [n], \forall z \in \{0,1\}^n \mbox{, z verifying } \forall j \in P_u,\, z_j = y_j, \mbox{ we have } f_u(y) = y_u =f_u(z)$.
\end{remark}

\begin{proof}
Let $u$ be a vertex in $[n]$. $f(u)$ only depends of the incoming arcs in $u$, so it only depends of $p_u$, which in turn depends on its parents. By induction, $f_u(y)$ only depends of $P_u$, and so, if $f_u(y) = y_u$ in $G$, then $f_u(y) = y_u$ in $\gp{u}$. \qed
\end{proof}

\section{Formalisation of the BN Reprogramming with Permanent Perturbations}
\label{sec:bn-reprogramming}

Given two fixed points $x$ and $y$ of Boolean network $f$, we want to identify
sets of nodes, referred to as Reprogramming Determinants (RDs), that when changed in $x$ enable to switch to $y$. As our theorems rely on the differences between the fixed points, we chose to focus on fixed points solely. Further work will extend, if possible, these theorems and algorithms to all kind of attractors.
In the scope of this paper, by "change" we mean permanently set the vertex to a new fixed value. If we "change" $u$ to $1$ (resp. $0$), then $f_u(x) =1$ (resp. $0$) for all $x$.
When switching to $y$ (by changing $I$) is possible, we have two cases : it
either means that $y$ is reachable from $x_{[x_I = y_I]}$ (existential
reachability, Def.~\ref{def:er}), or that $y$ is the only reachable fixed point
from $x_{[x_I = y_I]}$ (inevitable reachability, Def.~\ref{def:ir}).
These are two different approaches that we will both consider. To remove the temporal aspect, we make all the changes at the same time (hence $x_{[x_I = y_I]}$, otherwise an order should be visible), and only watch if $y$ is reachable. This also means that there is no indication of how long it takes for $y$ to be reached.

\begin{definition}[Existential Reachability]
\label{def:er}
With the boolean network F, a function $ER_F$ can be defined as $ER_F :
\powerset{\powerset{[n]}}$, with $ER_F(x,y) \mapsto v$ where $v$ is the set of all minimal vertex sets $I$ such that $x_{[x_I=y_I]} \rightarrow^* y$.
\end{definition}

\begin{definition}[Inevitable Reachability]
\label{def:ir}
Similarly, a function 
$IR_F : \powerset{\powerset{[n]}}$ can be defined as $IR_F(x,y) \mapsto w$ where $w$ is the set of all minimal vertices sets $I$ such as $\forall z \in \{0,1\}^n,\, x_{[x_I=y_I]} \rightarrow^* z \Rightarrow z \rightarrow^* y$.
\end{definition}

These two functions will give different results, and have different meanings, as shown in the examble below.

\begin{example}
\label{ex:reprog}
Let us consider the BN $f$ of Fig.\ref{fig:sbn} and its transition graph reproduced
in Fig.\ref{fig:sbn-tg}.
$f$ has 4 fixed points: $0000, 0001, 1100$ and $1101$.
Let $x=0000$ and $y=1100$. 
Fixing the node $\{1\}$ to $1$ in $x$ makes $y$ reachable :
$1100$ (=$y$) is reachable from $x_{[x_1=1]}=1000$ with the Boolean network $f'$ defined by
$f'_1(x) = 1$ and $f'_2=f_2$, $f'_3=f_3$, $f'_4=f_4$.
The transition graph of $f'$, considering the first node being active,
corresponds to the left part of the transition graph in Fig.\ref{fig:sbn-tg}.
One can then remark that $y$ is not the only fixed point reachable: from $1000$, $1101$ is also reachable.
If we also fix the node $\{4\}$ to $0$, $y$ is the only reachable fixed point from $x_{[x_1=1, x_4=0]}$
in the Boolean network $f''$ such that $f''_1(x) = 1$, $f''_2=f_2$, $f''_3=f_3$,
and $f''_4(x)=0$.

Therefore, with the previous definitions,
$\{1\} \in ER_F(0000,1100)$ but
$\{1\} \notin IR_F(0000,1100)$;
and
$\{1,4\}\in IR_F(0000,1100)$ but
$\{1,4\}\notin  ER_F(0000,1100)$.
Moreover, we also have $\{1,2\}$ and $\{1,3\}\in IR_F(0000,1100)$.
\begin{figure}
\begin{center}
\begin{tikzpicture}
\node[draw, align=left] (text) at (-4,-1)	{$f_1(x) = x_1 $ \\ $ f_2(x) = x_1 $ \\ $ f_3(x) = x_1 \wedge \neg x_3 $ \\ $ f_4(x) = x_3 \vee x_4$} ;
\node[main node]		(a) 					{1} ;
\node[main node]		(b) [below right of=a]	{2} ;
\node[main node]		(c) [below left of=a]	{3} ;
\node[main node]		(d) [below of=c]	{4} ;
\path 	(a) edge[blue, loop above] (a)
		(a) edge	[blue, ->]			(b)
		(a) edge	[blue, ->]			(c)
		(b) edge[red, -|]			(c)
		(c) edge[blue, ->]			(d)
		(d) edge[blue, loop below] (d) ;
\end{tikzpicture}
\end{center}
\caption{A BN of dimension $4$} \label{fig:sbn}
\end{figure}
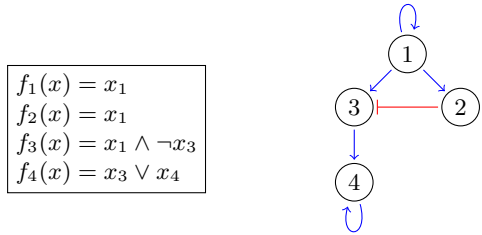

\begin{figure}
\begin{center}
\begin{tikzpicture}[node distance=1.3cm]

\node	(1010)	{1010} ;
\node	(1110)	[right of=1010]	{1110} ;
\node	(1000)	[below of=1010, label={[red]left:$\bar{x}^{\{1\}} =$}]	{1000} ;
\node[draw, align=center]	(1100)	[red, right of=1000, label={[red]below:$y$}]	{1100} ;

\node	(phantom0)	[above left of=1010]	{} ;

\node	(1011)	[above left of=phantom0]	{1011} ;
\node	(1111)	[right of=1011]	{1111} ;
\node	(1001)	[below of=1011]	{1001} ;
\node	(1101)	[red, right of=1001]	{1101} ;

\node	(phantom1)	[right of=1110] {};

\node	(0010)	[right of=phantom1] {0010} ;
\node	(0000)	[red, below of=0010, label={[red]below:$x$}]	{0000} ;
\node	(0110)	[right of=0010]	{0110} ;
\node	(0100)	[right of=0000]	{0100} ;

\node	(phantom2)	[above left of=0010]	{} ;

\node	(0011)	[above left of=phantom2]	{0011} ;
\node	(0111)	[right of=0011]	{0111} ;
\node	(0001)	[red, below of=0011]	{0001} ;
\node	(0101)	[right of=0001]	{0101} ;

\path	(0010) edge[->]	(0000)
		(0010) edge[bend right, ->]	(0011)
		(0011) edge[->]	(0001)
		(0100) edge[->]	(0000)
		(0101) edge[->]	(0001)
		(0110) edge[->]	(0100)
		(0110) edge[->]	(0010)
		(0110) edge[bend right, ->]	(0111)
		(0111) edge[->]	(0011)
		(0111) edge[->]	(0101)
		(1000) edge[->]	(1010)
		(1000) edge[->]	(1100)
		(1001) edge[->]	(1101)
		(1001) edge[->]	(1011)
		(1010) edge[bend right, ->]	(1011)
		(1010) edge[->]	(1110)
		(1011) edge[->]	(1111)
		(1110) edge[->]	(1100)
		(1110) edge[bend right, ->]	(1111)
		(1111) edge[->]	(1101) ;

\end{tikzpicture}
\end{center}
\caption{Transition graph of the BN in Fig.\ref{fig:sbn}
\label{fig:sbn-tg}}
\end{figure}
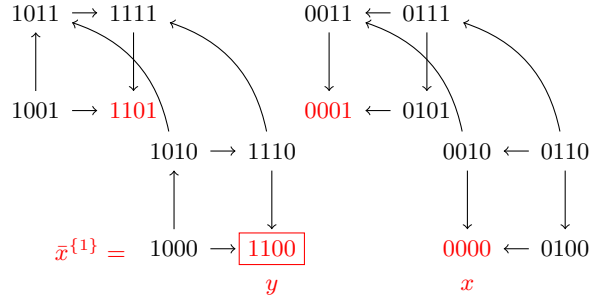
\end{example}

\section{Reprogramming Determinants and the SCCs of the Interaction Graph}
\label{sec:rd-scc}

In this section, we show the link between the RDs and the Strongly Connected
Components (SCCs) of the interaction graph of the Boolean network $f$.
Our results make the assumption that all the attractors of $f$ are fixed points
(no cyclic attractors).

\subsection{SCC Ordering}
To switch from $x$ to $y$, we want to change the value of each vertex $u$ that has different values for $x$ and $y$ ($x_u \neq y_u$) and to prevent each vertex $v$ that verifies $x_v = y_v$ from changing value. We know that changing the value of a vertex can have an impact on other vertices, but we also know that it will only impact its descendants. 

So, if a vertex has a different value in $x$ and $y$ but none of its ancestors do, then it is necessary to change this vertex. So, to know which vertices need to be changed first, the best way is to order them, with a topological order for example. 
Of course, if there are loops, an order is impossible to determine, we have to reduce all SCCs	 to single "super-vertices" to achieve it.
In the remaining of this paper, we will consider SCCs which contain at least
one positive cycle, because they are known to change between fixed points (Theor.\ref{thm:thomas}), we call $\SCC$ the SCC set that contains all such SCCs.
Reducing the graph to its SCCs makes possible to rank them from $1$ to $k$ with
any topological order, noted $\prec$:
for all $i,j\in[k], j > i\Rightarrow \SCC_j \not\prec \SCC_i$.

Let $C_0$ be the set $\{ \SCC_i \in \SCC \mid \nexists \SCC_j,\, \SCC_j \prec \SCC_i\}$, and recursively define slices $C_K = \{ \SCC_i \in (\SCC \setminus \bigcup_{l \in \{1,..,K-1\}}C_l) \mid \nexists \SCC_j,\, \SCC_j \prec \SCC_i\}$. Given the definition of the slices, for all topological orders, the slice set will be the same. The slices are numbered from $1$ to $c$.

From this order, we know which SCCs need to be impacted, still, SCCs ranked lower in the hierarchy
need not be impacted by the change in their ancestors (see ex.\ref{ex:blockswitch}) The relation $\prec$ only gives an order to make the changes, from which one can determine if further changes are needed.

\begin{example}
Showing that only using the topological order is not sufficient.
\begin{figure}
\begin{center}
\begin{tikzpicture}
\node[draw, align=left] (text) at (-4,0)	{$f_1(x) = \neg x_2$ \\ $f_2(x) = \neg x_1$ \\ $f_3(x) = x_1 \vee x_2$ \\ $f_4(x) = x_2 \wedge \neg x_3 $ \\ $f_5(x) = x_4 \vee x_5$} ;
\node[main node]		(c) 						{3} ;
\node[main node]		(a) [below left of=c]	{1} ;
\node[main node]		(b) [below right of=c]	{2} ;
\node[main node]		(d) [right of=c]			{4} ;
\node[main node]		(e)	[right of=d]			{5} ;
\path 	(a) edge[red, bend left, -|]	(b)
		(b) edge	[red, bend left, -|]	(a)
		(a) edge	[blue, ->]			(c)
		(b) edge[blue, ->]			(c)
		(c) edge[red, -|]			(d)
		(b) edge[blue, ->]			(d)
		(d) edge[blue, ->] 		(e)
		(e) edge[blue, loop right]	(e) ;
\end{tikzpicture}
\end{center}
\caption{BN preventing changes in the lower SCC}
\label{ex:blockswitch}
\end{figure}
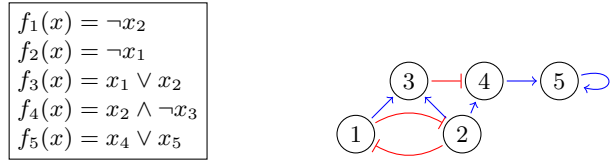

Any algorithm that only used the topological order without computing the reachable fixed points would not suffice, as the example from Fig.\ref{ex:blockswitch} shows :
the switch from the fixed point $01100$ to $10101$ would be computed by just modifying $\{1\}$, but in fact $\{4\}$ will always be fixed at $0$, because $\{4\}$ is always inhibited by $\{3\}$, so $\{5\}$ needs to be changed too.
\label{ex:blockswitcht}
\end{example}

\subsection{SCC Filtering}

Whether we want $y$ to be the only reachable attractor, or merely to be one of potential several such attractors, the ordering from the previous part is the same, but the filtering will differ. 
\begin{theorem}
\label{theor:ancestrality}
If a vertex $u$ such as $x_u \neq y_u$ and $u$ is not in a positive cycle, then modifying $u$'s ancestors is sufficient to modify $u$. \\ More generally, to switch from $x$ to $y$, modifying only those strongly connected components that contain at least a positive cycle is sufficient.
\end{theorem}
\begin{proof}
Let $u$ be a vertex such that $x_u \neq y_u$ and $u$ does not lie in a positive cycle. If $u$ is in a negative cycle, the incoming arc from the cycle is irrelevant : given that $x$ and $y$ are fixed points and that $u$ has a distinct value in each, the negative cycle does not change $u$'s value. Given that $u$ is not in a positive cycle, $u$ is not in a SCC (or not relevant if it is in a negative cycle). That means that none of the ancestors are descendants of $u$. Let $z$ be the state where all of $P_u$ ($u$'s ancestors) have the same value that in $y$. By the remark from Sect.\ref{sec:background}, for all $v \in \gp{u}$, we have $f_v(z)=z_v=y_v$. So, either $f_u(z)=y_u$, and the theorem is proven, either $f_u(z) \neq y_u$, then, by Theor.\ref{thm:thomas}, $u$ is in a positive cycle, contradiction.
\qed
By recursion over the first part, modifying all the SCCs that contain positive cycles so their vertices have the same value as in $y$ modifies all their children, and then all the children of their children, and so on, until the whole graph has the same values as $y$. \qed

\end{proof}

Selecting the SCCs will differ with the two methods. It relies on the same base, searching the higher SCC that should have its values modified and that is not already selected. "Modified" means that all the values of the SCC are fixed to their values in $y$. The set of the selected SCCs is $\SC$.

\subsection{SCC Filtering for Existential Reachability}

We consider the RDs for the BN reprogramming with Existential Reachability.
We give an algorithm to identify different sets of SCCs for which the mutation
in the initial fixed point ensure the reachability of the target fixed point.
We will prove that the identified combination of SCCs is complete and minimal.

Basically, the algorithm reviews linearly the SCC slices according to $\prec$
and adds the minimal combinations of SCCs to $\SC$ that are different in $y$ and
the fixed points reachable from $x_{[x_{\SC} = y_{\SC}]}$:
\begin{enumerate}
\item $\SC := \emptyset$
\item For $i$ ranging from $1$ to $c$:
\begin{itemize}
\item $T := \emptyset$
\item $\forall s \in P(C_i)$ such that $s$ minimal\\ $\exists z \in \{0,1\}^n,\, z_{C_i \setminus s} = y_{C_i \setminus s},\,  x_{[x_I = y_I | I \in  \SC]} \to^*z$, $T := T \cup s.$
\item $\SC := \SC \bar{\times} T$.
\end{itemize}
\end{enumerate}
With $\bar{\times}$ being a product and union : for a set $I$ of subsets $I_1,..,I_k$ and a set $J_1,..,J_l$, this product $\bar{\times}$ is defined by : $I \bar{\times} J = \{I_1 \cup J_1, .., I_1 \cup J_l, I_2 \cup J_1, ...., I_k \cup J_l\}$

\paragraph*{Complexity :}
In the worst case, the above algorithm perform $c\times 2^l$ reachability checks
(PSPACE-complete \cite{ChengEP95}), where $l$ is the size of the largest slice.

\paragraph*{Existence of a solution and proof of correctness :}
Forcing all SCCs of such problem that differ on $x$ and $y$ to have the same value as in $y$ is one solution. In the worst case, that is what the algorithm will find. Since the algorithm tests reachability, and a solution exists, it will find one.

\begin{example}
We apply the algorithm on the BN of Fig.\ref{fig:bn-slices}
with $x=00000$ and $y=11011$.
\begin{figure}
\begin{center}
\begin{tikzpicture}
\node[draw, align=left] (text) at (-4,-1)	{$f_1(x) = x_1 $ \\ $ f_2(x) = x_1 $ \\ $ f_3(x) = x_1 \wedge \neg x_2 $ \\ $ f_4(x) = x_3 \vee x_4$\\ $f_5(x) = x_2 \vee x_5 $} ;
\node[main node]		(a) 					{1} ;
\node[main node]		(b) [below right of=a]	{2} ;
\node[main node]		(c) [below left of=a]	{3} ;
\node[main node]		(d) [below of=c]		{4} ;
\node[main node]		(e) [below of=b]		{5} ;
\node[draw, align=left, minimum width=2.5cm, minimum height=0.5cm] at (0,0) [magenta, label={[magenta]right:$C_1$}] {} ;
\node[draw, align=left, minimum width=2.5cm, minimum height=0.5cm] at (0,-1.7) [magenta, label={[magenta]right:$C_2$}] {} ;
\path 	(a) edge[blue, loop above] (a)
		(a) edge	[blue, ->]			(b)
		(a) edge	[blue, ->]			(c)
		(b) edge[red, -|]			(c)
		(c) edge[blue, ->]			(d)
		(b) edge[blue, ->]			(e)
		(e) edge[blue, loop below]	(e)
		(d) edge[blue, loop below] 	(d) ;
\end{tikzpicture}
\end{center}
\caption{BN of dimension $5$ (left) with its interaction graph (right).
Slices are enclosed in boxes. $C_1=\{\{1\}\}$, $C_2=\{\{4\},\{5\}\}$.
\label{fig:bn-slices}}
\end{figure}
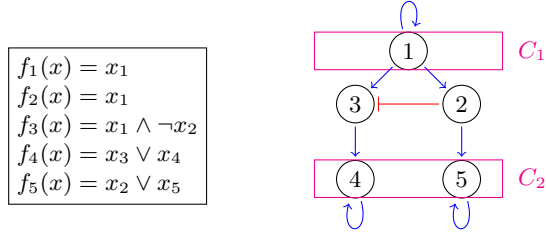

\begin{enumerate}
\item 
$\SC := \emptyset$
\item 
$C_1$: $s$ minimal $\Leftrightarrow \, s = \{1\}$
\item
$\SC := \SC \bar{\times} \{1\} = \{\{1\}\}$
\item
$C_2$: $s$ minimal $\Leftrightarrow \, s = \emptyset$ \footnote{with the path $10000 \rightarrow 10100 \rightarrow 10110 \rightarrow 11110 \rightarrow 11111$ (and the fixed point is the next step, $\rightarrow 11011$ but there is no need to go further than $11111$.)}
\item 
$\SC := \SC \bar{\times} \emptyset = \{\{1\}\}$.
\end{enumerate}
\end{example}

We now prove the completeness of the algorithm and the minimality of the
returned sets of SCCs (Theorem~\ref{thm:er-algo})
and that any RDs in $ER(x,y)$ is spans only and necessarily in one of the set of
SCCs identified by the algorithm (Theorem~\ref{thm:er-scc}).

\begin{theorem}
$\SC$ only contains minimal SCC sets, and $\SC$ is complete.
\label{thm:er-algo}
\end{theorem}

\begin{proof}
\textbf{Minimality :}
Inside every slice, the SCCs are totally independant one another. Moreover, given the order exploiting, we can deduce that the sum of the minima on each slice is the minimum on the whole graph.
\qed	

\textbf{Completeness :}
Let $I$ be a minimal SCC set such as $x_{[x_J = y_J | J \in I]} \to^* y$, then, for every slice $C_i$, $I \cap C_i$ is minimal, since once all the SCCs in a slice can be changed to the way they are in $y$, we can always choose the path that allows this change. Hence $I \in \SC$.
\qed
\end{proof}

\begin{theorem}
$\forall c \in ER(x,y),\, \exists I \in \SC,\, \forall u \in c,\, \exists scc \in I,\, u \in scc$.
\label{thm:er-scc}
\end{theorem}
\begin{proof}
Let $c$ be a vertex set in $ER(x,y)$ and $u$ one of the vertices. If $u \not \in \SCC$, then $c$ is not minimal, by Theorem \ref{theor:ancestrality}. If for all $I \in \SC$, $u$ is in $o \in (\SCC \setminus I)$ then there exists a path such that changing $o$'s ancestors makes $o$'s change possible, and the ancestors need to be changed as well, by construction of $I$. So $c \setminus u$ would have the same effect, and $c$ would not be minimal. If $u \not \in o$, then there exists $I \in \SC$ and $scc \in I$, such as $u \in scc$. \qed
\end{proof}

\subsection{SCC Filtering for Inevitable Reachability}

We now give an algorithm to identify a set of SCCs for which the mutation in the
initial fixed point is sufficient to ensure the \emph{Inevitable Reachability}
of the target fixed point.

The algorithm computes all reachable fixed points from $x$ with the SCCs in $\SC$ modified, and find the one, $z$, that has the lower SCC (in the ranking given by $\prec$) in which a vertex $u$ is such that $z_u \neq y_u$. As we are looking for all reachable fixed points, this will always return the same SCC (even if the order is only partial), thus allowing the algorithm to be deterministic. We add this SCC to $\SC$, and repeat until $y$ is the only reachable fixed point.

\begin{enumerate}
\item $\SC := \emptyset$
\item While $\exists z \in \FP,\,z \neq y,\, x_{[x_I = y_I | I \in  \SC]} \to^* z$
\begin{itemize}
\item $\SC := \SC \cup \{\SCC_i\}$, with\\ 
$i = min_{a \in \{1,..,k\}}(a \mid \exists z \in \FP,\, z_{\SCC_a} \neq y_{\SCC_a},\, x_{[x_I = y_I | I \in \SC]} \to^* z) $ 
\end{itemize}
\end{enumerate}

If two (or more) SCCs $A$ and $B$ are such that they are differently ordered in two distinct orders, then $A$ has no influence on $B$ and neither has $B$ on $A$. Then, the algorithm will select both SCCs if neither are impacted by the previous changes, so the order does not matter.

\paragraph*{Existence of a solution and proof of correctness :}
A solution is to fix all the SCCs of the graph to their value in $y$. Since there exists a solution and the algorithm tests if $y$ is the only reachable point, and follows the order given by $\prec$, it will end and find a solution.

\paragraph*{Complexity :}
Computing all fixed points reachable is PSPACE-complete \cite{CHJPS14-CMSB}.
It is used $k$ times (number of SCCs) in the worst case.

\begin{example}
We apply the algorithm on the BN of Fig.\ref{fig:bn-ir}.
with $x=00000$ and $y=11011$.
Starting from $\SC := \emptyset$, the
only reachable fixed point is $(0)00(0)(0)$ (the SCCs from $\SCC$ are parenthesized).
The smallest SCC $o$ such as $x_{[x_{\SC}=y_{\SC}], o} \neq y_o$ is $\SCC_1$, so
$\SC := \emptyset \cup \{\SCC_1\} = \{\SCC_1\}$.
The reachable fixed points from $x_{[x_I = y_I | I \in \SC]}=10000$ are now: $(1)10(1)(1)$ and $(1)10(0)(1)$.
The smallest SCC $o$ such that $x_{[x_{\SC}=y_{\SC}], o} \neq y_o$ is $\SCC_3$.
We set $\SC := \SC \cup \SCC_3 = \{\SCC_1, \SCC_3\}$
and we obtain that the only reachable fixed point from $x_{[x_I = y_I | I \in \SC]}=10010$ is $(1)10(1)(1)$ which is $y$.
So the algorithm stops.
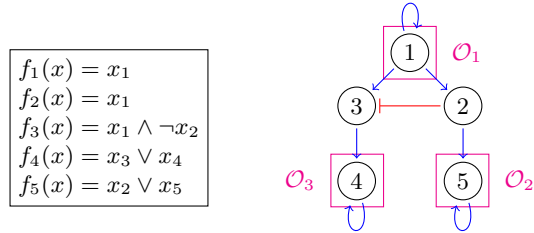
\begin{figure}
\begin{center}
\begin{tikzpicture}
\node[draw, align=left] (text) at (-4,-1)	{$f_1(x) = x_1 $ \\ $ f_2(x) = x_1 $ \\ $ f_3(x) = x_1 \wedge \neg x_2 $ \\ $ f_4(x) = x_3 \vee x_4$\\ $f_5(x) = x_2 \vee x_5 $} ;
\node[main node]		(a) 					{1} ;
\node[main node]		(b) [below right of=a]	{2} ;
\node[main node]		(c) [below left of=a]	{3} ;
\node[main node]		(d) [below of=c]		{4} ;
\node[main node]		(e) [below of=b]		{5} ;
\node[draw, align=left, minimum width=0.7cm, minimum height=0.7cm] at (0,0) [magenta, label={[magenta]right:$\SCC_1$}] {} ;
\node[draw, align=left, minimum width=0.7cm, minimum height=0.7cm] at (-0.7,-1.7) [magenta, label={[magenta]left:$\SCC_3$}] {} ;
\node[draw, align=left, minimum width=0.7cm, minimum height=0.7cm] at (0.7,-1.7) [magenta, label={[magenta]right:$\SCC_2$}] {} ;
\path 	(a) edge[blue, loop above] (a)
		(a) edge	[blue, ->]			(b)
		(a) edge	[blue, ->]			(c)
		(b) edge[red, -|]			(c)
		(c) edge[blue, ->]			(d)
		(b) edge[blue, ->]			(e)
		(e) edge[blue, loop below]	(e)
		(d) edge[blue, loop below] 	(d) ;
\end{tikzpicture}
\end{center}
\caption{BN of dimension 5 (left) with its interaction graph (right) on which
the SCCs containing positive cycles ($\SCC$) are boxed.
\label{fig:bn-ir}}
\end{figure}
\end{example}

\begin{theorem}
$\SC$ is minimal.
\end{theorem}

\begin{proof}
If a set $S_1$ exists such that $S_1$ has a lower cardinal than $\SC$ and modifying $S_1$ makes $y$ the only reachable point, then we can reduce $S_1$ to a subset of $\SC$. Let $s$ be a SCC in $\SC \setminus S_1$, thus there exists a fixed point $z$ such that $z_s \neq y_s$ and by construction of $\SC$, $z$ is reachable from $x$ modified by $S_1$. \qed
\end{proof}

We remark that, contrary to the case of Existential Reachability, the RDs for
Inevitable Reachability of the target fixed point are not necessarily in SCCs
containing positive cycles.
Indeed, in Ex.\ref{ex:reprog}, we showed that $IR_F(x,y)$ can refer to nodes
that do not belong to $\SCC$ (such as the node $2$ for the BN of
Fig.\ref{fig:sbn}). But we can also remark that if a RD $v$ is not in a SCC containing a positive cycle, then $x_v = y_v$.

\begin{theorem}
$\forall v \in IR(x,y),\, x_v \neq y_v \Rightarrow \exists scc \in \SCC,\, v \in scc$
\end{theorem}

\begin{proof}
Let $v \in IR(x,y)$ such that for all $scc \in \SCC$, $v \not \in scc$. By Theor.\ref{thm:thomas}, if $v$ is such that $x_v \neq y_v$, then modifying the SCC in $\SCC$ is enough to modify $v$. But $v \in IR(x,y)$ and $IR(x,y)$ is minimal, thus $x_v = y_v$. 
\qed
\end{proof}

\section{Identifying Determinants within SCCs}
\label{sec:enumeration}

We know that modifying all the SCCs selected is enough to switch from $x$ to $y$, but to reduce the genes selected, we could try to modify only some of the vertices to achieve the same result. But, as dynamics are involved, there could be unwanted changes (or wanted and unpredicted changes, in the case where we want $y$ to be reachable) in the descendants.

An idea could be to select the feedback vertex set of the SCC : by fixing the
vertices from this set, we effectively destroy every circle, thus the only
reachable state of the SCC is the one having the same values as $y$. This, however, does not solve the problem : in Ex.\ref{ex:dynamicsPb}, $\{1\}$ is the feedback vertex set, and we still have the same issue. Moreover, it miss some of the possible solutions (modifying $\{2\}$ or $\{3\}$ could work to change the whole SCC in Ex.\ref{ex:dynamicsPb}) or even dismiss the best solution (in Ex.\ref{ex:dynamicsPb}, changing $\{3\}$ makes $y$ the only reachable fixed point and solves the issue).

\begin{example}
Illustration of the problem with dynamics.
\begin{figure}
\begin{center}
\begin{tikzpicture}
\node[draw, align=left] (text) at (-4,0)	{$f_1(x) = \neg x_3 \wedge \neg x_2$ \\ $f_2(x)= \neg x_1$ \\ $f_3(x) = \neg x_1$ \\ $f_4(x) = x_2 \wedge \neg x_1 \wedge \neg x_3$ \\ $f_5(x) = x_4 \vee x_5$} ;
\node[main node]		(a) 									{1} ;
\node[main node]		(b) [above right of=a]				{2} ;
\node[main node]		(c) [below right of=a]				{3} ;
\node[main node]		(d) [right of=a]						{4} ;
\node[main node]		(e)	[right of=d]						{5} ;
\path 	(a) edge[red, bend left, -|]	(b)
		(b) edge	[red, bend left, -|]	(a)
		(a) edge[red, bend left, -|]	(c)
		(c) edge	[red, bend left, -|]	(a)
		(a) edge	[red, -|]			(d)
		(b) edge[blue, ->]			(d)
		(c) edge[red, -|]			(d)
		(d) edge[blue, ->]			(e)
		(e) edge[blue, loop right]	(e) ;
\end{tikzpicture}
\end{center}
\caption{BN of dimension 5 (left) and its interaction graph (right)
\label{fig:bn-complex}}
\end{figure}
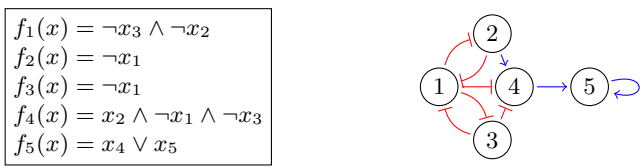

We decide that $x=10000$ and $y=01100$, and $01101=z$, those are all fixed points. Let's suppose we want $y$ to be the only fixed point reachable. The algorithm will see that if the whole first SCC is modified $\{1,2,3\}$, $y$ is the only reachable fixed point. 
It could pick $\{1\}$ to be modified, but instead of $00000 \rightarrow 00100 \rightarrow 01100$, we can have $$00000 \rightarrow 01000 \rightarrow 01010 \rightarrow 01011 \rightarrow 01111 \rightarrow 01101$$

This leads to $z$ being reachable by only modifying $\{1\}$, and so the algorithm would be wrong.
\label{ex:dynamicsPb}
\end{example}

This leaves to two kinds of approaches : either a way to modify the SCC so that it does not impact its descendants can be found, either we need to select the vertices to be modified in the SCC as a intermediary step in the process, and redesign $\SC$ as a list of vertices instead of a list of SCCs.

By exploiting the results of the preceding section, we show an algorithm to
compute a set of RDs which guarantees the Inevitable Reachability of the target
fixed point.
The algorithm recursively picks a vertex $u$ in the lowest SCC in the order given by $\prec$ in $\SCC$, and
modify its associated function to become the constant value $y_u$.
The interaction graph of the resulting Boolean network is a sub-graph of the
initial interaction graph, where all the input edges of the node $u$ have been
removed.
Hence, the SCC $\SCC_1$ is split in the new interaction graph.
If necessary, another vertex can be picked in the lowest SCC in the new
interaction graph:

\medskip

\noindent
RecursiveAlgorithm($f$, $rd$) :
\begin{itemize}
\item If $\exists z \in \FP,\, x_{[x_{rd} = y_{rd}]} \to^* z$ then :
\begin{itemize}
\item $res=\emptyset$
\item
$i = min_{a \in \{1,..,k\}}(a \mid \exists z \in \FP,\, z_{\SCC_a} \neq y_{\SCC_a},\, x_{[x_I = y_I | I \in \SC]} \to^* z) $ 
\item For all $u \in \SCC_i$ :
\begin{itemize}
\item $g := f$ with $g_u := y_u$
\item $res := res \cup$ RecursiveAlgorithm($g$, $rd \bar{\times} \{u\}$)
\end{itemize}
\item return $res$
\end{itemize}
\item else :
\begin{itemize}
\item return $rd$
\end{itemize}
\end{itemize}

Remark that the algorithm always find at least one solution: if the target fixed
point is not the only reachable fixed point, then there is at
least one positive cycle (and hence a SCC) which has a different state (and
hence will be selected by our algorithm).

\begin{example}
Applied to the BN of Fig.\ref{fig:bn-complex} with $x=10000$ and $y=01100$, the
above algorithm returns, for instance, the RD $\{2,5\}$:
indeed, $\{2\}$ belongs to $\SCC_1$.
When fixing $f_2=1$, the new interaction graph has two SCCs with positive
cycles: $\{1,3\}$ and $\{5\}$.
From the state $11000$, two fixed points are reachable:
$01100$, $01101$.
Hence, because the SCC $\{1,3\}$ has the same values than in $y$ in those two fixed
points,
the next vertex in picked in the SCC $\{5\}$.
Finally, from the state $11001$, $y$ is the only reachable fixed point.
\end{example}

\section{Discussion}
\label{sec:discussion}

This paper provides the first formal characterization of the Reprogramming Determinants (RDs) for
switching from one fixed point to another in the scope of the asynchronous dynamics of Boolean
networks.

In the case of reprogramming with existential reachability of the target fixed point, we
prove that all the possible minimal RDs modify nodes in particular combinations
of SCCs of the interaction graph of the Boolean network.
We give an algorithm to determine exactly those combinations of set of nodes.
Our characterizations rely on the verification of reachability properties.

In the case of reprogramming with inevitable reachability of the target fixed
point, we show that the RDs are not necessarily in SCCs.
However, we provide an algorithm which identifies RDs that guarantee the
inevitable reachability by picking nodes in appropriate SCCs.
The algorithm relies on the enumeration of reachable fixed points.

One of the main limitation of our algorithms is the numerous reachability checks it needs to perform.
Future work will consider methods and data structures for factorizing the exploration of the Boolean
network dynamics.

The present work considered only permanent mutations: when a node is mutated, it is assumed it
keeps its mutated value forever (its local Boolean function becomes a constant function).
Considering temporary mutations, i.e., where the local Boolean function of mutated nodes is restored
after some time, is a challenging research direction:
one should determine the ordering and the duration of mutations, and the set of candidate mutations
is \emph{a priori} no longer restricted to connected components, as it is the case for permanent
mutations.

\bibliographystyle{plain}
\bibliography{CellReprog_Methods_IG}


\end{document}